\documentclass[runningheads]{llncs}
\usepackage[margin=1in]{geometry}
\usepackage{amsmath}
\usepackage{array,xspace,multirow,hhline,graphicx,xcolor,tikz,colortbl,tabularx,amsmath,amssymb,amsfonts}
\usetikzlibrary{shapes}
\usetikzlibrary{arrows}
\usetikzlibrary{decorations.pathreplacing,calligraphy}

\usepackage{caption,tabularx}

\usepackage[numbers]{natbib}

\makeatletter 
\renewcommand\@biblabel[1]{#1} 
\makeatother

\usepackage{subcaption}
\usepackage{collcell}
\usepackage{hhline}
\usepackage{pgf}
\usepackage{pgfplots}
\usepgfplotslibrary{external}
\usepackage{wrapfig}
\newcommand\eat[1]{}
\usepackage{pgfplots}
\usepackage{bbm}
\usepackage{verbatim,ifthen}
\usepackage{pifont}
\usepackage{ifthen}
\usepackage{pgflibraryshapes}

\usepackage{nicefrac}
\usetikzlibrary{arrows}
\usepackage{ltxtable}
\usepackage{longtable}
\usepackage[normalem]{ulem}
\usepackage{varioref}

\usepackage{calc}
\newsavebox\CBox
\newcommand\hcancel[2][0.5pt]{%
	\ifmmode\sbox\CBox{$#2$}\else\sbox\CBox{#2}\fi%
	\makebox[0pt][l]{\usebox\CBox}%
	\rule[0.5\ht\CBox-#1/2]{\wd\CBox}{#1}}

\tikzset{
	jumpdot/.style={mark=*,solid},
	excl/.append style={jumpdot,fill=white},
	incl/.append style={jumpdot,fill=black},
	rexcl/.append style={jumpdot,color=red,fill=white},
	rincl/.append style={jumpdot,fill=black,color=red},
}
	\newcommand{\stkout}[1]{\ifmmode\text{\sout{\ensuremath{#1}}}\else\sout{#1}\fi}
	
	\usepackage[ruled]{algorithm2e} 
	
	\SetAlFnt{\small}
	\SetAlCapFnt{\small}
	\SetAlCapNameFnt{\small}
	\SetAlCapHSkip{0pt}
	\IncMargin{-\parindent}

	\usepackage{amsmath}
	\usepackage{mathtools}

	\newcommand{\med}{\operatorname{med}}

	\definecolor{gray(x11gray)}{rgb}{0.75, 0.75, 0.75}
	\def\colorModel{rgb} 
	
\DeclareMathOperator*{\argmin}{arg\,min}
	\newcommand\ColCell[1]{
		\pgfmathparse{#1<0.5?1:0}  
		\ifnum\pgfmathresult=0\relax\color{white}\fi
		\pgfmathsetmacro\compA{1-#1}      
		\pgfmathsetmacro\compB{1-#1/1.5} 
		\pgfmathsetmacro\compC{1}      
		\edef\x{\noexpand\centering\noexpand\cellcolor[\colorModel]{\compA,\compB,\compC}}\x #1
	} 
	\newcolumntype{E}{>{\collectcell\ColCell}m{0.5cm}<{\endcollectcell}}  
	
\begin{document}
		\title{Facility Location Games with Scaling Effects} 
		
		%
		%
		
		
		\author{Yu He \and Alexander Lam \and Minming Li
		}
		\authorrunning{..}
		
		\institute{City University of Hong Kong\\
			\email{yuhe32-c@my.cityu.edu.hk, \{alexlam, minming.li\}@cityu.edu.hk}}
		
		\maketitle              
		\begin{abstract}
We take the classic facility location problem and consider a variation, in which each agent's individual cost function is equal to their distance from the facility multiplied by a scaling factor which is determined by the facility placement. In addition to the general class of continuous scaling functions, we also provide results for piecewise linear scaling functions which can effectively approximate or model the scaling of many real world scenarios. We focus on the objectives of total and maximum cost, describing the computation of the optimal solution. We then move to the approximate mechanism design setting, observing that the agents' preferences may no longer be single-peaked. Consequently, we characterize the conditions on scaling functions which ensure that agents have single-peaked preferences. Under these conditions, we find a characterization of continuous, strategyproof, and anonymous mechanisms, and compute the total and maximum cost approximation ratios achievable by these mechanisms.
\end{abstract}
		\section{Introduction}
		The facility location problem is a widely studied problem, with decades of literature spanning the fields of operations research, economics, and computer science. In the classic variant of the problem, agents are located on a unidimensional interval, and we are tasked with finding an ideal location to place a facility which serves these agents. The agents have single-peaked preferences for the facility location, as they incur a cost equal to their distance from the facility. This problem models many real-world single-peaked social choice and preference aggregation problems. The geographical placement of public facilities such as libraries and parks is an immediate example \citep{Miya01}, and the problem can also be applied to choose a political or economic outcome \citep{BoJo83,FFG16}, or to decide how to divide a budget \citep{FPPV21}. Furthermore, the one-dimensional setting can be generalized to a multi-dimensional setting with an $L_1$ distance metric \citep{PSS92,GoHa23}, or extended to a network setting \cite{ScVo02}.

        In our setting, we assume that the agents' locations are private information, and take a strategyproof mechanism design approach to the classic facility location problem. We also add an additional dimension of complexity: each agent's cost is scaled by an external factor corresponding to the facility's location on the domain. Specifically, there is a (continuous) \emph{scaling function} which maps the facility location to a positive scaling factor which is multiplied by each agent's distance to calculate their individual costs. A local minimum of the scaling function implies that the facility is particularly effective when placed at this point. The primary motivation for our model is the placement of a cell-phone tower, which is most effective when it is on top of a hill, and less effective as it is placed lower on the hill. Each signal broadcasted or received by the radio tower has a failure probability depending on its elevation, and this averages out over time to be the facility’s effectiveness in a multiplicative sense. Other examples include the proximity of the facility to a river or public transportation, or of a warehouse to a shipping port. Similarly, placing the facility at a local maximum is relatively ineffective (such as placing the cell phone tower at the bottom of a valley, or a public facility next to a garbage dump), and the scaling function improves up to an extent as the facility moves away from this point.

\begin{figure}
    \centering   
            \includegraphics[scale=0.56]{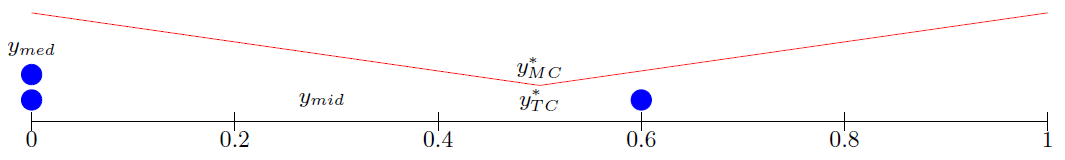}
\caption{Facility location example where the agents (\textcolor{blue}{$\bullet$}) are located at $(0,0,0.6)$, and the scaling function (represented by the red lines) is $2-3y$ for $y\leq 0.5$, and $3y-1$ for $y>0.5$. The median and midpoint facility placements are at $y_{med}=0$ and $y_{mid}=0.3$ respectively, whilst the optimal facility placements in terms of total and maximum cost are at $y^*_{TC}=y^*_{MC}=0.5$.}
    \label{fig:example}
\end{figure}
        
        In our paper, we consider the general class of continuous scaling functions, as well as piecewise linear scaling functions due to their simplicity in modelling real-world problems, and because they can approximate more complex scaling functions. To demonstrate our model, consider the example in Figure~\ref{fig:example}, where there are two agents at $0$ and one agent at $0.6$, and the scaling function is $2-3y$ for $y\leq 0.5$, and $3y-1$ for $y>0.5$. Under the classic model, the median and midpoint mechanisms are known to minimize total and maximum cost, respectively. In our example, they place the facility at $0$ and $0.3$ respectively, resulting in $0.6\times 2=1.2$ total cost and $0.33$ maximum cost. However, placing the facility at the scaling function's minimum of $0.5$ is optimal for both objectives, leading to $0.55$ total cost and $0.25$ maximum cost.

        The goal of our paper is to design strategyproof mechanisms for our scaled facility location problem, that have ideal approximations of the optimal total or maximum cost. These two objectives are standard measures of efficiency and egalitarian fairness in facility location mechanism design \citep{PrTe13}. A key consequence of the scaling function is that the agents' preferences may no longer be single-peaked, and thus we cannot immediately leverage existing characterizations of strategyproof mechanisms under single-peaked preferences (e.g., \citep{Moul80,BoJo83}). We therefore consider the sufficient and necessary conditions on the scaling function such that the agents' preferences are single-peaked, and design reasonable strategyproof mechanisms for such scaling functions.

        \paragraph{Contributions}
        \begin{itemize}
            \item We initiate the study of facility location mechanism design where the facility has a degree of effectiveness depending on where it is placed, and the agents' costs are scaled by this level of effectiveness. Our results focus on continuous scaling functions, as well as the smaller space of piecewise linear scaling functions.
            \item We give computational results on the optimal facility placement for the objectives of total cost and maximum cost.
            \item We show that the agents' preferences may not be single-peaked in our setting, and characterize the conditions on scaling functions such that agents have single-peaked preferences.
            \item When scaling functions meet the conditions for agents to have single-peaked preferences, we characterize the space of continuous mechanisms satisfying strategyproofness and anonymity as \emph{phantom mechanisms with scaling}.
            \item We find the lower bound on the approximation ratios achieved by phantom mechanisms with scaling, and propose such mechanisms with approximation ratios matching this lower bound. Table~\ref{tbl} summarizes some of our approximation ratio results.
        \end{itemize}
  \begin{table}[]
\centering
\bgroup
\def\arraystretch{1.5}
\begin{tabular}{lllll}
\cline{2-5}
           & \multicolumn{2}{l}{General Pref.}                                                               & \multicolumn{2}{l}{Single-Peaked Pref.}                                      \\ \cline{2-5} 
           & \multicolumn{2}{l}{\def\arraystretch{1}\begin{tabular}[c]{@{}l@{}}Continuous and \\ Piecewise Linear $q$\end{tabular}} & Continuous $q$ & Piecewise Linear $q$ \\ \hline
Total Cost & \multicolumn{2}{l}{$r_q$ (Thm.~\ref{thm:totalcostgen})}                                                                       & $e$ (Cor.~\ref{cor: e})         & $(1+\frac{1}{k})^k$ (Cor.~\ref{cor: 2})                                                            \\ \hline
Max. Cost  & \multicolumn{2}{l}{$2r_q$ (Thm.~\ref{thm:maxcostgenLB})}                                                                      & $2e$  (Cor.~\ref{cor: 2e})      &$2(1+\frac{1}{k})^k$ (Cor.~\ref{cor: 4})                                                            \\ \hline
\end{tabular}
\egroup
\caption{Phantom mechanism approximation ratio results for total and maximum cost. Agents have single-peaked preferences when the scaling function $q$ meets certain conditions, resulting in phantom mechanisms being strategyproof. The term $r_q$ denotes the ratio between the max. and min. values of $q$, and $k$ denotes the number of line segments in the piecewise linear $q$. All results are tight in the sense that they are a lower bound for all phantom mechanisms, and that there exists a phantom mechanism with a matching approximation ratio.}
\label{tbl}
\end{table}
\newpage
		\section{Related Work}
		The facility location problem can be interpreted as a special case of preference aggregation when agents have symmetric and single-peaked preferences over a continuous range of outcomes, in which there has been decades of research in strategyproof aggregation rules. For general preferences, the famous Gibbard-Satthertwaite theorem implies that a surjective and strategyproof rule must be dictatorial \citep{Gibb73,Satt75}. This impossibility result vanishes when we consider the restricted domain of single-peaked preferences, in which there have been numerous characterizations of strategyproof rules. Most notably, under this restricted domain, \citet{Moul80}'s seminal paper characterizes the space of strategyproof and anonymous rules as Phantom mechanisms which output the median of the agent locations and $n+1$ constant `Phantom' values. This characterization was later extended to the domain of symmetric and single-peaked preferences by \citet{BoJo83} and \citet{MaMo11a}. Other characterizations of strategyproof rules have been proposed for similar settings by \citep{BaJa94,NePu06,NePu07a,JLPV23}.
		
		The paper most similar to our work considers the facility location model with entrance fees \citep{MXBK23}, in which agents incur, in addition to their distance from the facility, an additive cost which depends on the facility placement. In this model, the authors give strategyproof mechanisms and compute their approximation ratios with respect to the total and maximum cost objectives. The key difference with our model is that our agents experience a multiplicative scaling factor to their cost, whilst their agents experience an additional additive cost. Our work and \citep{MXBK23} are inspired by the line of research on approximate mechanism design without money, initially proposed for the facility location problem by \citet{PrTe13}. In their paper, the authors compute the worst-case ratio between the performance of strategyproof mechanisms and welfare-optimal mechanisms. There have since been many papers applying this approach to variations of the facility location problem, such as when there are proportionally fair distance constraints \citep{ALLW21}, when facilities have capacity limits \citep{ACLP20,ACLL+20}, or when the facility is obnoxious (i.e., agents prefer to be far away from the facility) \citep{CYZ13}. A recent survey of related work is given by \citet{CFLLW21}.

		\section{Model}
		We have a set of agents $N=\{1,\dots,n\}$, where agent $i$ has location $x_i$ on the domain\footnote{Although we consider the unit interval domain, our results can be scaled and shifted to any compact interval on $\mathbb{R}$. We consider a bounded domain so that a single-peaked linear scaling function does not take negative values.} $X:=[0,1]$. Although we consider the space of potentially non-anonymous mechanisms, we assume for simplicity that agent locations are ordered such that $x_1\leq \dots \leq x_n$. This does not affect the nature of our results. A scaling function $q: X\rightarrow \mathbb{R}_{>0}$ gives the effectiveness of a facility.\footnote{If there is a point $y$ where $q(y)=0$, then it is trivial to place the facility at that point.} We also assume that the scaling function is continuous. As we will show in Proposition~\ref{prop:discon}, continuity of $q$ is required for the optimal facility location to be well-defined. 
  
        Denoting $\mathcal{Q}$ as the space of all scaling functions, a \emph{continuous}\footnote{For mechanisms, we define continuity with respect to the agents' locations. Formally, we say a mechanism is continuous if $\forall \mathbf{x}\in X^n: \forall \epsilon>0:\exists \delta>0:\forall \mathbf{x}'\in X^n, \forall q\in \mathcal{Q}: \lVert \mathbf{x}-\mathbf{x}'\rVert_1<\delta \implies \lVert f(q,\mathbf{x})-f(q,\mathbf{x}')\rVert_1<\epsilon$.} facility location mechanism $f: \mathcal{Q}\times X^n\rightarrow X$ maps the agent location profile $\boldsymbol{x}=(x_1,\dots,x_n)$ to the location of a facility $y$. We define the cost incurred by agent $i$ as its distance from the facility multiplied by the scaling factor: $c_i(q,y):= q(y)|y-x_i|$. Finally, we define the total cost of an instance as $TC(q,y,\boldsymbol{x}):=\sum_i c_i(q,y)=q(y)\sum_i |y-x_i|$, and the maximum cost of an instance as $MC(q,y,\boldsymbol{x}):=\max_i c_i(q,y)=q(y)\max_i |y-x_i|$. Respectively, we denote the optimal facility location which minimizes the total cost (resp. maximum cost) as $y^*_{TC}$ (resp. $y^*_{MC}$).
		
		It is typically ideal for the mechanism output to be independent of the agents' labelling, so we are concerned with mechanisms that satisfy \emph{anonymity}.
        \begin{definition}
        A mechanism $f$ is \emph{anonymous} if its output does not change when the agents' labels are permuted.
        \end{definition}

		\section{Computing the Optimal Solution}
  In this section, we investigate the properties of the optimal facility locations for total cost $y^*_{TC}$ and maximum cost $y^*_{MC}$. We begin by justifying our assumption that the scaling function $q$ must be continuous, as a discontinuity may result in the optimal solution being not well-defined.
  \begin{proposition}\label{prop:discon}    
			There exists a scaling function with a discontinuity and an agent location profile such that the optimal facility locations for total cost $y^*_{TC}$ and maximum cost $y^*_{MC}$ are not well-defined.
		\end{proposition}		
		\begin{proof}
			Consider the location profile $\textbf{x} = (0,0.5,1)$ and the scaling function  
			$
			q(y) = \begin{cases}
				1, & y \in [0,0.5],\\
				0.5, & y \in (0.5,1].
			\end{cases}
			$
			The optimal facility locations for total cost and maximum cost are the right limit $y^*\rightarrow 0.5^+$, but not at $y^*=0.5$, so the optimal solutions are not well-defined. \qed
		\end{proof}
  \subsection{Total Cost}
        We now give computational results on the optimal facility location for total cost, beginning with the class of linear scaling functions.
		\begin{proposition}\label{prop:optlinear}
			If $q$ is a linear function, the optimal facility location for total cost $y_{TC}^*$ is either on one of the agents' locations, or the minimum value of $q$.
		\end{proposition}
		\begin{proof}
			Since $q$ is a linear function, it is of the form $q(y) = ay+b$, where $a\in \mathbb{R}$ and $b\in \mathbb{R}_{>0}$. If $a=0$, then $q(y)$ is constant, and thus the optimal facility location is on the median agent.\footnote{If there are an even number of agents, we break the tie in favour of the leftmost agent in the optimal interval.}
			
			Due to symmetry, it suffices to consider the case where the scaling function is increasing, i.e., $a>0$. From this assumption, we know that the optimal solution cannot lie to the right of $x_n$.
   
			Since $q$ is increasing with $y$, its minimum value is on $y = 0$. We denote $x_0:=0$, and we know the optimal solution $y^*$ is in the interval $[x_0,x_n]$. Suppose the optimal solution $y^*$ is in the interval $[x_{n_1},x_{n_1+1}]$, where $n_1 \in [0,n-1]$ denotes the number of agents on the left-hand side of $y^{*}$. Also, let $n_2 = n-n_1$ denote the number of agents on the right-hand side of $y^{*}$.
			
			Recall that the total cost corresponding to $y^*$, $TC(q,y^*,\textbf{x})$ is $q(y^{*})\sum_i|y^*-x_i|$.
			Considering the term $\sum_i |y^*-x_i|$, we have
			\begin{align*}
				\sum_{i=1}^n |y^*-x_i| & = \sum_{i=1}^{n_1}(y^*-x_i)+\sum_{i=n_1+1}^{n}(x_i-y^*)\\
				& = (n_1-n_2)y^* + \sum_{i=n_1+1}^{n}x_i - \sum_{i=1}^{n_1}x_i.
			\end{align*}
			We remark that $\sum_i |y^*-x_i|$ can also be written as a linear function of $y^*$: $\sum_i |y^*-x_i| = a'y^{*}+m$ with $a' := n_1-n_2$, and $m := \sum_{i=n_1+1}^{n}x_i-\sum_{i=1}^{n_1}x_i$. We can therefore express the total cost as the following function of $y^*$.
			\begin{align*}
				TC(q,y^*,\textbf{x}) & = q(y^{*})\cdot \sum_i|y^*-x_i|\\
				& = (ay^*+b)(a'y^*+m)\\
				& = aa'(y^*)^2+ (am+a'b)y^*+bm.
			\end{align*}
			Now we know that the total cost function can either be a linear function or a quadratic function of $y^{*}$, depending on whether $a'=0$, $a'<0$, or $a'>0$ (recalling that $a>0$). There are three cases: either the total cost function is a linear function ($a'=0$), a quadratic function which opens downwards ($a'<0$), or a quadratic function which opens upwards ($a'>0$). For the first two cases, it is clear that the minimum value of the total cost function must lie on an endpoint of the interval $[x_{n_1},x_{n_1+1}]$. It remains to consider the last case ($a'>0$), where the total cost function $ q(y^*)\sum_i|y^*-x_i| = aa'y^*+ (am+a'b)y^*+bm$ opens upwards.
			Consider the total cost function on the larger domain of $\mathbb{R}$. Since $a'=(n_1-n_2)>0$, we know there exists $z_1\in \mathbb{R}$ such that $a'z_1+m=0$. Also, since $a > 0$, there exists $z_2\in \mathbb{R}$ such that $q(z_2)=0$. As a result, we see that the total cost function is equal to zero at either one or two points on $\mathbb{R}$; we now show that both of these points lie to the left of $x_{n_1}$.
			
            For the term $a'y^{*}+m$, we know that $a'z_1+m =0$, and thus 
            $z_1=-\frac{m}{a'}=\frac{\sum_{i=n_1+1}^{n}x_i-\sum_{i=1}^{n_1}x_i}{n_2-n_1}.$
Consider the numerator $\sum_{i=n_1+1}^{n}x_i-\sum_{i=1}^{n_1}x_i$. We know there are $n_2$ agents between $(x_{n_1},x_n]$, implying $\sum_{i=n_1+1}^{n}x_i > n_2 \cdot x_{n_1}$. Similarly, we know there are $n_1$ agents between $[0,x_{n_1}]$, implying $\sum_{i=1}^{n_1}x_i \leq n_1 \cdot x_{n_1}$. Combining these inequalities gives us
 $$\sum_{i=n_1+1}^{n}x_i-\sum_{i=1}^{n_1}x_i>n_2 x_{n_1} - n_1 x_{n_1}>0.$$
			As we have assumed $a'>0$, we know $n_2-n_1<0$ and thus
   $z_1=\frac{\sum_{i=n_1+1}^{n}x_i-\sum_{i=1}^{n_1}x_i}{n_2-n_1}<x_{n_1}.$

   Now consider $z_2$, which satisfies $q(z_2)=0$. Clearly, $z_2<0$ as we must have $q(y)>0$ for $y\in [0,1]$ and $a>0$, implying $z_2<x_{n_1}$.
   
			Since $z_1, z_2<x_{n_1}$, the total cost function must be strictly increasing in $[x_{n_1},x_{n_1+1}]$, so its minimum value must lie on $x_{n_1}$, which is either an agent location or the minimum value of $q$. \qed
		\end{proof}
		
		We can generalise this result to continuous, \emph{piecewise} linear scaling functions.		
		\begin{theorem}
			Let the scaling function $q$ be a continuous, piecewise linear function. The facility location minimizing total cost $y_{TC}^*$ is either on one of the agents' locations or on a local minimum of $q$.
		\end{theorem}
		\begin{proof}
			First, consider each line segment on $q$, which is represented by a linear function $q_j$ on the sub-domain $[c_j,d_j]$, where $c_j,d_j\in X$. Let $y_j^*:=\arg \min_{y_j\in [c_j,d_j]} TC(q,y_j,\boldsymbol{x})$ denote the optimal facility placement on $[c_j,d_j]$, and $e_j:=\arg \min_{y\in \{c_j,d_j\}}q(y)$ be the minimum point of $q$ on $[c_j,d_j]$. By Proposition~\ref{prop:optlinear}, we must have $y_j^*\in \{e_j,x_1,\dots,x_n\}$. Therefore, considering all of the $q_j$ functions, the optimal solution $y^*$ must lie at one of the agents' positions or at a local minimum of $q_j$. \qed
		\end{proof}
		\begin{corollary}
			When $q$ is a continuous, piecewise linear function with a constant number of local minima, the optimal facility location can be found in linear time by simply iterating over the agent locations and the local minima of the scaling function.
		\end{corollary}

        As we assume agent locations are private information, there is a concern that agents may misreport their locations to unfairly attain a better facility placement. It is therefore ideal to implement a \emph{strategyproof} mechanism, which does not incentivize agents to lie about their locations.
		\begin{definition}[Strategyproofness]
			A mechanism $f(\cdot)$ is \emph{strategyproof} if for every agent $i \in N$, we have, for every scaling function $q$ and agent locations $x_i'$, $\boldsymbol{x}_{-i}$ and $x_i$,
			$$
			c_i(q,f(x_i,\textbf{x}_{-i})) \leq c_i(q,f(x_i',\textbf{x}_{-i})).
			$$
		\end{definition}
		However, we find that the optimal solution is not strategyproof.
		\begin{proposition}
			The optimal facility placement which minimizes total cost is not strategyproof, even if the scaling function is a linear function and there are only 3 agents.
		\end{proposition}
		\begin{proof}
			Suppose the scaling function is
			$q(y)=1-\frac{4}{5}y$
			and the agent location profile is $x=(0.3,0.4,0.7)$. The optimal facility placement is $y=0.4$, which leads to a cost of $0.204$ for agent $x_3$. However, if agent $x_3$ misreports $x_3'=0.755$, then the facility placement is instead $y'=1$, which leads to a cost of $0.06$. Hence the optimal facility placement is not strategyproof. \qed
		\end{proof}
        \subsection{Maximum Cost}
        We now consider the optimal facility location for maximum cost. It is well-known that in the classic setting, placing the facility at the midpoint of the leftmost and rightmost agents is optimal for maximum cost but not strategyproof \citep{PrTe13}. This implies that $y^*_{MC}$ is not strategyproof in our setting. However, it is easily computable in constant time when $q$ is a linear function, as there are only three possible solutions.
            \begin{proposition}\label{prop:optmaxlinear}
                If $q$ is a linear function, the optimal facility location for maximum cost satisfies $y^*_{MC}\in \{0,\frac{x_1+x_n}{2},1\}$.
            \end{proposition}
            \begin{proof}
                Suppose $q$ is of the form $q(y) = ay+b$, and recall that the maximum cost of the instance is $MC(q,y,\boldsymbol{x}) = \max_i c_i(q,y) = \max_i (ay+b)|y-x_i|$. Due to symmetry it suffices to consider $a>0$.

                Since $a>0$, we know that $y^*_{MC}\notin (x_n,1]$, as such a facility placement has a strictly higher cost for all agents than a facility on $x_n$. We also know that $y^*_{MC}\notin (\frac{x_1+x_n}{2},x_n]$. This is because for every facility placement on this interval, there is a `symmetric' facility placement $y\in [x_1,\frac{x_1+x_n}{2})$ where the maximum (unscaled) distance is the same, but the scaling function value is lower. We therefore know that $y^*_{MC}\in [0,\frac{x_1+x_n}{2}]$. When the facility satisfies $y\leq \frac{x_1+x_n}{2}$, then we also know that agent $n$ incurs the highest cost, as it has the highest unscaled distance from the facility.

                We now know that if $a>0$, the maximum cost is $MC(q,y,x)=c_n(q,y) = (ay+b)(x_n-y)$, which is a quadratic function of $y$ that opens downwards. Therefore the maximum cost is minimized when $y=0$ or $y=\frac{x_1+x_n}{2}$. By symmetry, if $a<0$, agent $1$ incurs the maximum cost of $MC(q,y,x)=c_1(q,y) = (ay+b)(y-x_1)$, which is minimized when $y=1$ or $y=\frac{x_1+x_n}{2}$. We therefore have $y^*_{MC}\in \{0,\frac{x_1+x_n}{2},1\}$. \qed        
            \end{proof}
            
            We can also similarly show that when computing $y^*_{MC}$ for continuous scaling functions, we only need to consider the optimal facility location with respect to the leftmost and rightmost agents.
            \begin{proposition}
                For continuous scaling functions, the optimal facility location for maximum cost satisfies $$y^*_{MC}\in \left\{\argmin_{y\in[0,\frac{x_1+x_n}{2}]}c_n(q,y), \argmin_{y\in[\frac{x_1+x_n}{2},1]}c_1(q,y)\right\}.$$
            \end{proposition}
            \begin{proof}
                If $y\in[0,\frac{x_1+x_n}{2}]$, then agent $n$ receives the highest cost in the instance, similar to the proof of Proposition~\ref{prop:optmaxlinear}. A symmetric argument shows that if $y\in[\frac{x_1+x_n}{2},1]$, then agent $1$ receives the highest cost in the instance. Therefore $y^*_{MC}$ is either on the minimum value of $c_n(q,y)$ on $y\in[0,\frac{x_1+x_n}{2}]$, or the minimum value of $c_1(q,y)$ on $y\in[\frac{x_1+x_n}{2},1]$. \qed 
            \end{proof}
            
		We next investigate whether there exist non-trivial strategyproof mechanisms for this setting.
		\section{Achieving Single-Peaked Preferences}
		When agents have single-peaked preferences along a compact domain, \citet{Moul80} characterized the set of anonymous and strategyproof mechanisms as \emph{phantom mechanisms}, which place the facility at the median of the $n$ agent locations and $n+1$ constant `phantom' points.
		\begin{definition}[Phantom Mechanism]
			Given \textbf{x} and $n+1$ constant values $0\leq p_1\leq p_2 \leq \cdots \leq p_{n+1} \leq 1$, a \emph{phantom mechanism} places the facility at  $\text{med}\{x_1,\cdots,x_n,p_1,\cdots,p_{n+1}\}.$
		\end{definition}
		However, we find that in our setting, phantom mechanisms are not necessarily strategyproof. Even the well-studied median mechanism, a phantom mechanism with $\lceil \frac{n+1}{2} \rceil$ phantoms at $0$ and $\lfloor \frac{n+1}{2} \rfloor$ phantoms at $1$, is not strategyproof.
		 \begin{proposition}
		     The median mechanism is not strategyproof.
		 \end{proposition}
		\begin{proof}
			Consider the agent location profile $\boldsymbol{x}=(0.5,0.7,0.9)$ and the scaling function
			$q(y)=\begin{cases}
				6-10y, & y\leq 0.5,\\
				10y-4, & y>0.5.
			\end{cases}
			$
   
			The median mechanism places the facility at $y=0.7$, and it leads to a cost of $0.6$ for agent $x_3$. However if the agent at $x_3$ misreports to $x_3' = 0.5$, then the median facility placement will be $0.5$, which leads to a lower cost of $0.4$ for the agent at $x_3$. Hence the median mechanism is not strategyproof. \qed
		\end{proof}
        In fact, if there are at least two unique phantom locations, the phantom mechanism is not strategyproof.

        \begin{theorem}\label{thm:phantomnotsp}
			Every phantom mechanism with $n+1$ phantoms on at least two unique phantom locations is not strategyproof.
		\end{theorem}
		
		\begin{proof}
   Suppose the $n+1$ phantoms are located on $k\in \{2,\dots,n+1\}$ unique points $(a_1,a_2,\cdots,a_k)$. We consider an agent profile such that all the agents are in the interval $[a_1,a_2]$. Denoting the number of phantoms on $a_1$ as $n^{*}$, the facility is placed on the $(n+1-n^{*})$th agent under this profile. 
   Now suppose the scaling function is
   $$q(y)= \begin{cases}
            3-\frac{a_1+1}{a_2+1}-\frac{4}{a_1+a_2}y, & y \in [0,\frac{a_1+a_2}{2}],\\
                      \frac{8}{a_2+1}y-\frac{5a_1+3a_2}{a_2+1}, & y \in [\frac{a_1+a_2}{2},1].
   \end{cases}
   $$
   
   \textbf{Case 1} ($n+1-n^{*} \in \{2,\dots,n\}$):
   With all agents located in $[a_1,a_2]$, we also specifically let $x_{n-n^{*}} = a_1,x_{n+1-n^{*}} = \frac{3a_1+a_2}{4}$, and $x_{n+2-n^{*}} = \frac{a_1+a_2}{2}$.
   In such case, the mechanism places the facility at $x_{n+1-n^{*}}$, and the agent at $x_{n-n^{*}}$ incurs cost 
            \begin{align*}
            c_{(n-n^{*})}(q,f(\boldsymbol{x})) &= q(\frac{3a_1+a_2}{4})\cdot|\frac{3a_1+a_2}{4}-a_1|\\
            &= (3-\frac{a_1+1}{a_2+1}-\frac{4}{a_1+a_2}\cdot\frac{3a_1+a_2}{4})(\frac{a_2-a_1}{4})\\
            &= (1-\frac{1}{2}(\frac{a_1+1}{a_2+1}+\frac{2a_1}{a_1+a_2}))(\frac{a_2-a_1}{2})\\
            &>(1-\frac{a_1+1}{a_2+1})(\frac{a_2-a_1}{2}).
             \end{align*}
            The last inequality is due to $\frac{a_1+a_1}{a_2+a_1} < \frac{a_1+1}{a_2+1}$. 
            If $x_{n-n^{*}}$ misreports to $x_{n-n^{*}}'=\frac{a_1+a_2}{2}$, then $f(\boldsymbol{x}') = \frac{a_1+a_2}{2}$, and their new cost becomes $c_{(n-n^{*})}(q,f(\boldsymbol{x}'))  = (1-\frac{a_1+1}{a_2+1})(\frac{a_2-a_1}{2})<c_{(n-n^{*})}(q,f(\boldsymbol{x}))$.
            
    \textbf{Case 2} ($n+1-n^{*} = 1$):
   Consider the agent location profile $\boldsymbol{x}$ where $x_1 = \frac{a_1+3a_2}{4}$ and $x_n = a_2$. Here, we have $f(\boldsymbol{x}) = x_{n+1-n^{*}} = x_1 = \frac{a_1+3a_2}{4}$. Therefore the cost for agent $x_n$ is 
   \begin{align*}
       c_n(q,f(\boldsymbol{x})) &= q(\frac{a_1+3a_2}{4})\cdot|\frac{a_1+3a_2}{4}-a_2|\\
       &=  (\frac{8}{a_2+1}\cdot\frac{a_1+3a_2}{4}-\frac{5a_1+3a_2}{a_2+1})|\frac{a_1+3a_2}{4}-1|\\
       &= (\frac{3a_2-3a_1}{a_2+1})(\frac{a_2-a_1}{4})\\
       &= \frac{3}{2}\cdot(1-\frac{a_1+1}{a_2+1})(\frac{a_2-a_1}{2})\\
       &> (1-\frac{a_1+1}{a_2+1})(\frac{a_2-a_1}{2}).
   \end{align*}    
   However if $x_n$ misreports to $x_n'=\frac{a_1+a_2}{2}$, then $f(\boldsymbol{x}') = \frac{a_1+a_2}{2}$. Agent $x_n$ will then incur cost \begin{align*}
       c_n(q,f(\boldsymbol{x}'))&=q(\frac{a_1+a_2}{2})\cdot|\frac{a_1+a_2}{2}-a_2|\\
       &= (\frac{8}{a_2+1}\cdot\frac{a_1+a_2}{2}-\frac{5a_1+3a_2}{a_2+1})(\frac{a_2-a_1}{2})\\
       &= (1-\frac{a_1+1}{a_2+1})(\frac{a_2-a_1}{2})\\
       &<c_n(q,f(\boldsymbol{x})).
   \end{align*}
   Therefore, a phantom mechanism which does not place all phantoms at the same position is not strategyproof. \qed
  \end{proof}

    From this theorem, we arrive at the following corollary. Note that the constant mechanism is equivalent to the phantom mechanism with all $n+1$ phantoms at the same location.
        \begin{corollary}
        The constant mechanism (which places the facility at some constant $c\in X$) is the only anonymous and strategyproof phantom mechanism in our setting.
        \end{corollary}
  
		This is because the scaling function may cause the agent's preference to no longer be single-peaked. To show this, we first define each agent's utility as the negative of their cost ($u_i(q,y) := -c_i(q,y)$), so that their preference (with respect to cost) is single-peaked if and only if their utility is single-peaked. We therefore use the terms `utility functions' and `preferences' interchangeably throughout the paper.
		
		\begin{theorem}
			An agent's preference may not be single-peaked, even if the scaling function is a linear function.	
		\end{theorem}
		\begin{proof}
			Consider the scaling function $q(y) = y+0.01$ and $x_1=0.5$. The derivative of agent $1$'s utility is $u_1'(q,y)=2y-0.49$ when $y < 0.5$, and $u_1'(q,y)=-(2y-0.49)$ when $y > 0.5$. For $y\in[0,0.245]$, $u_1'(q,y) \leq 0$. For $y\in [0.245,0.5]$, $u' \geq 0$. Lastly, for $y\in [0.5,1]$, $u_1'(q,y) < 0$. Therefore $u_1(q,y)$ is not single-peaked. \qed
		\end{proof}

However, we can characterize the conditions on the scaling functions which ensure that agents have single-peaked preferences. These conditions hold as long as $q$ is continuous, even if it is not differentiable at a countable set of points $D$.
   \begin{theorem}\label{thm:charcontnotdiff}
                    If $q$ is continuous but not differentiable at a countable set of points $D:=\{j|q'(j)\, \text{does not exist}\}$, the agents' preferences are guaranteed to be single-peaked if and only if $q(y)-|q'(y)| \geq 0$ for all $y\in [0,1]\backslash D$.
                \end{theorem}
                \begin{proof}
                   Recall that the utility function for agent $x_i$ is
			$$u_i(q,y) = \begin{cases}
				q(y)(y-x_i), & y \leq x_i,\\
				q(y)(x_i-y), & y \geq x_i.
			\end{cases}$$
                The derivative of $u$ with respect to $y$ is
			$$
			u_i'(q,y)=\begin{cases}
                    q(y)+q'(y)(y-x_i), & y \in [0,x_i)\backslash D, \\
                    \text{Not differentiable}, & y\in D\cup\{x_i\},\\
				-(q(y)+q'(y)(y-x_i)), & y \in (x_i,1]\backslash D.
			\end{cases}
			$$
                $(\implies)$ The utility function is single-peaked, so for each agent $i$, there exists a value $k_i \in (0,1)$ such that $u_i'(q,y) \geq 0$ with $y \in[0,k_i)$, and $u_i'(q,y) \leq 0$ with $y \in(k_i,1]$. Since agents receive zero cost when the facility is placed at their location, we know the utility function's `peak' lies on their location, and therefore $k_i=x_i$.
                
               As a result, we have $q(y)+q'(y)(y-x_i) \geq 0$ for $y \in [0,x_i)\backslash D$, and $-q(y)+q'(y)(y-x_i) \leq 0$ for $y\in (x_i,1]\backslash D$. This can also be written as $q(y)+q'(y)(y-x_i) \geq 0$ for all $y\in [0,1]\backslash (\{x_i\}\cup D)$. Having $x_i \in [0,1]$ implies that $(y-x_i)\in[-1,1]$, meaning that our condition is $q(y)-|q'(y)| \geq 0$ for all $y\in[0,1]\backslash D$.

                $(\impliedby)$ Suppose that a scaling function $q$ satisfies $q(y)-|q'(y)| \geq 0$ for all $y\in[0,1]\backslash D$. We know that for every agent $x_i$, $q(y)+q'(y)(y-x_i) \geq 0$ for all $y\in [0,1]\backslash (\{x_i\}\cup D)$. The utility function is increasing for $y<x_i$, and decreasing for $y>x_i$. \qed
                \end{proof}
   By using this theorem, we can characterize the conditions on linear scaling functions which result in agents having single-peaked preferences.
                \begin{corollary}
                    If $q$ is a linear function of the form $q(y) = ay +b$, the agents' preferences are guaranteed to be single-peaked if and only if $b\geq a$ for $a > 0$, and $b \geq -2a$ for $a<0$.\\
                \end{corollary}

                We also observe that a sufficiently large constant $c$ can be added to any continuous scaling function to ensure that every agent has a single-peaked preference. This can be explained by the scaling function becoming `closer' to a constant scaling function.

                We can also apply Theorem~\ref{thm:charcontnotdiff} to piecewise linear scaling functions to characterize the conditions for agents to have single-peaked preferences.
                \begin{proposition} \label{prop:piece}
                    If $q$ is a continuous piecewise linear function where each line segment is of the form $q_j(y) = a_jy+b_j$, the agents' preferences are guaranteed to be single-peaked if and only if for each line segment, $b_j \geq -2a_j$ when $a_j<0$, and $b_j \geq a_j$ when $a_j>0$.
                \end{proposition}
                \begin{proof}
                     From Theorem~\ref{thm:charcontnotdiff}, we know that if the $q$ function is continuous but not differentiable at some set of points $D$, the utility function is single-peaked if and only if $q(y)$ satisfies $q(y)-|q'(y)| \geq 0$ for all $y\in[0,1]\backslash D$. When $a_j<0$, the utility function is single-peaked if and only if $a_jy+b_j+a_j \geq 0 \Rightarrow b_j \geq -2a_jy$, and when $a_j>0$, the utility function is single-peaked if and only if $a_jy+b_j-a_j \geq 0 \Rightarrow b_j \geq a_j(y-1)$. Since $y\in[0,1]$, these conditions become $b_j \geq -2a_j$ and $b_j \geq a_j$, respectively. \qed
                \end{proof}

\section{Characterization of Strategyproof and Anonymous Mechanisms}
As we have characterized the conditions for the scaling function to guarantee single-peaked agent preferences, it may seem immediate to apply the results by \citet{Moul80} to characterize strategyproof and anonymous mechanisms under these conditions as phantom mechanisms. However, recall that the mechanism takes the scaling function as an additional input, meaning that the constant value $p_1,\dots,p_{n+1}$ may be dependent on the scaling function. Furthermore, the domain of single-peaked preferences induced by a scaling function meeting the key condition of Theorem~\ref{thm:charcontnotdiff} may be a strict subset of the domain of all arbitrary single-peaked preferences. As a result, the characterization by \citet{Moul80} does not immediately hold in our setting.

Nevertheless, we are still able to obtain a similar characterization of strategyproof and anonymous mechanisms in our setting when the scaling function guarantees single-peaked agent preferences, but we additionally require that the mechanism is continuous. We first define an adaptation of the phantom mechanism, in which each `constant' value admits the scaling function as input.

\begin{definition}[Phantom Mechanism with Scaling]
			Given an agent location profile $\mathbf{x}$, a scaling function $q$, and $n+1$ `phantom values' $\{p_i(q)\}_{i\in [n+1]}$ defined by $p_i: \mathcal{Q}\rightarrow [0,1]$, a \emph{phantom mechanism with scaling} places the facility at  $$\med\{x_1,\cdots,x_n,p_1(q),\cdots,p_{n+1}(q)\}.$$
		\end{definition}

We now present our characterization, which leverages a key result by \citet{BGSS24}.

\begin{theorem}\label{thm:charphant}
    When agents are guaranteed to have single-peaked preferences, a continuous mechanism is strategyproof and anonymous if and only if it is a phantom mechanism with scaling.
\end{theorem}
\begin{proof}
    $(\impliedby)$ Since the `phantom' values of a phantom mechanism with scaling are independent of the agents' locations, the mechanism is, given a scaling function, equivalent to a phantom mechanism with respect to any misreported agent locations or permutations of the agents' labellings. Thus, a phantom mechanism with scaling is strategyproof and anonymous when agents have single-peaked preferences. It remains to show the opposite direction.

    $(\implies)$ Let $\mathcal{U}^{SP}$ denote the set of all single-peaked preferences. Theorem~1 of \citep{BGSS24} states that for an arbitrary $\mathcal{U}\subseteq \mathcal{U}^{SP}$, a continuous mechanism $f$ is anonymous and strategyproof if and only if there exists $p_1,\dots,p_{n+1}$ in $[0,1]$ such that the facility is placed at
    $\text{med}\{x_1,\dots,x_n,p_1,\dots,p_{n+1}\}$. Note that $p_1,\dots,p_{n+1}$ must be independent of the agents' locations, but can depend on the preference domain $\mathcal{U}$. In our model, the set of single-peaked preferences $\mathcal{U}$ is induced by the scaling function, and hence $p_1,\dots,p_{n+1}$ can be dependent on the scaling function. It immediately follows that when agents have single-peaked preferences, a continuous mechanism which is strategyproof and anonymous must be a phantom mechanism with scaling. \qed    
\end{proof}
		\section{Approximation Ratio Results}
		In this section, we give results on the approximation ratio of certain mechanisms. The approximation ratio of a mechanism quantifies its worst-case performance for a specified objective function. 
  \begin{definition}[Total Cost Approximation Ratio]
			Given a mechanism $f$, its total cost approximation ratio is 
			$$
			\max\limits_{q\in \mathcal{Q},\boldsymbol{x}\in X^n}
			\frac{TC(q,f(q,\boldsymbol{x}),\boldsymbol{x})}{TC(q,y^*_{TC},\boldsymbol{x})}.
			$$
		\end{definition}
    The maximum cost approximation ratio is similarly defined by replacing the total cost expressions $TC(q,\cdot,\boldsymbol{x})$ and $y^*_{TC}$ with the corresponding maximum cost expressions $MC(q,\cdot,\boldsymbol{x})$ and $y^*_{MC}$.
    
         While we have defined the approximation ratio for the general space of scaling functions $\mathcal{Q}$, we additionally prove approximation ratio bounds for more restricted classes of scaling functions, such as piecewise linear scaling functions, and scaling functions such that agents are guaranteed to have single-peaked preferences.

  For simplicity of notation, we let $r_q:=\frac{\max_{y\in X}q(y)}{\min_{y\in X}q(y)}$ be the ratio between the highest and lowest values of the scaling function. As we will show, our approximation ratio results are typically a function of $r_q$. We show that when the scaling function is chosen such that agents are guaranteed to have single-peaked preferences, the value of $r_q$ is upper bounded by a constant.
     \begin{theorem}\label{thm:genspe}
      When agents are guaranteed to have single-peaked preferences, $r_q \leq e$.
  \end{theorem}
  \begin{proof}
      From Theorem~\ref{thm:charcontnotdiff}, we know that $q(y)-|q'(y)| \geq 0$ for all $y\in[0,1]\backslash D$, where $D$ denotes the points where $q(y)$ is not differentiable.
      
      First, consider the intervals where $q'(y) > 0$, and consequently $q$ satisfies $q(y) - q'(y) \geq 0$.
      We define a function $g(y):=e^yq(y)$, where $g'(y) = e^y(q(y)-q'(y))$. We know $g(y)$ is non-decreasing because $q(y) - q'(y) \geq 0$. Consider the points $y_1,y_2\in [0,1]$ where $y_1< y_2$. Since $g(y_1) \leq g(y_2)$, we know that $e^{y_1}q(y_1) \leq e^{y_2}q(y_2)$. Also, since $q(y)>0$ and $e^{y}>0$, we have $\frac{q(y_1)}{q(y_2)} \leq e^{y_2-y_1},$ which is maximized when $y_1=0$ and $y_2=1$. Therefore $r_q \leq e$. Furthermore, the function $q(y)=e^y$ satisfies $q(y)-|q'(y)| \geq 0$ and has $r_q=e$.
      
      Lastly, consider the intervals where $q'(y) < 0$, and consequently $q$ satisfies $q(y) + q'(y) \geq 0$.
      Similarly, we define $g(y):=e^{-y}q(y)$, where $g'(y) = e^{-y}(q'(y)-q(y))$. We know $g(y)$ is non-increasing because $q'(y)-q(y) \leq 0$. Consider the points $y_1,y_2\in [0,1]$ where $y_1< y_2$. Since $g(y_1) \geq g(y_2)$, we know that $e^{-y_1}q(y_1) \geq e^{-y_2}q(y_2)$. Also, since $q(y)>0$ and $e^{y}>0$, we have $\frac{q(y_2)}{q(y_1)} \leq e^{y_2-y_1}$, which is again maximized when $y_1=0$ and $y_2=1$. Therefore $r_q \leq e$. \qed
  \end{proof}
  When $q$ is a piecewise linear function, the upper bound of $r_q$ is dependent on how many line segment pieces the function has.
    \begin{theorem}\label{thm:piecelin}
        Suppose the scaling function $q$ is a piecewise linear function with $k$ line segment pieces. When agents are guaranteed to have single-peaked preferences, $r_q\leq (1+\frac{1}{k})^k$.
    \end{theorem}
    \begin{proof}
        Suppose for $j\in \{1,\dots,k\}$ that each line segment of $q$ is denoted as $q_j(y) = a_jy+b_j$.

         From Theorem~\ref{thm:charcontnotdiff}, for agents to always have single-peaked preferences, we require that $q(y)-|q'(y)| \geq 0$ for all $y\in[0,1]\backslash\{D\}$, where $\{D\}$ denotes the set of points where $q(y)$ is not differentiable. This implies that for each line segment, if $a_j > 0$, then $a_jy+b_j-a_j > 0 \implies b_j \geq (1-y)a_j$. Similarly, if $a_j < 0$, then $b_j \geq -(1+y)a_j$.

         We denote the endpoints of the scaling function's line segment pieces as the set $S = \{0,s_1,s_2\cdots,s_{k-1},1\}$. Consider a line segment piece $q_j= a_jy+b_j$, where $y\in[s_{j-1},s_{j}]$. We now find an expression for $r_{q_j}$, which is the ratio between the maximum and minimum values of the scaling function on $[s_{j-1},s_{j}]$.
         
         If $a_j > 0$, then $$r_{q_j} = \frac{q_j(s_{j})}{q_j(s_{j-1})} = \frac{a_js_j+b_j}{a_js_{j-1}+b_j}.$$ Since we require $b_j \geq (1-y)a_j$ on $y\in[s_{j-1},s_{j}]$, we know that $b_j \geq (1-s_{j-1})a_j$. Therefore $$r_{q_j}=\frac{a_js_j+b_j}{a_js_{j-1}+b_j} \leq \frac{a_js_j+(1-s_{j-1})a_j}{a_js_{j-1}+(1-s_{j-1})a_j} = (1+s_j-s_{j-1}).$$
         
         If $a_j < 0$, then $$r_{q_j} = \frac{q_j(s_{j-1})}{q_j(s_{j})} = \frac{a_js_{j-1}+b_j}{a_js_{j}+b_j}.$$ Since we require $b_j \geq -(1+y)a_j$ on $y\in[s_{j-1},s_{j}]$, we know that $b_j \geq -(1+s_{j})a_j$. Therefore $$r_{q_j}=\frac{a_js_{j-1}+b_j}{a_js_{j}+b_j} \leq \frac{a_js_{j-1}-(1+s_{j})a_j}{a_js_{j}-(1+s_{j})a_j} = (1+s_j-s_{j-1}),$$ and we know that for each $q_j$, $r_{q_j} \leq (1+s_j-s_{j-1})$.

         Since $q$ is piecewise linear, its global maxima and minima must be elements of $S$. We take an arbitrary global minimum $s_{min}:=\arg\min_{y\in X]}q(y)$ and global maximum $s_{max}:=\arg\max_{y\in X}q(y)$, where $s_{min},s_{max}\in S$, and suppose w.l.o.g. that $s_{min} < s_{max}$. It also suffices to assume that for each line segment $q_j$ between $s_{min}$ and $s_{max}$, we have $a_j>0$ (i.e. the function is strictly increasing in $[s_{min},s_{max}]$). This means that $r_q$ is the product of the $r_{q_j}$ between $s_{min}$ and $s_{max}$. 
         
         Finally, since $1 \leq r_{q_j} \leq (1+s_j-s_{j-1})$, and applying the arithmetic mean-geometric mean inequality, we have $r_q \leq \prod \limits_{i=1}^k (1+s_i-s_{i-1})\leq (1+\frac{1}{k})^k$. \qed
    \end{proof}
  When the number of line segments $k$ approaches infinity, the upper bound of $r_q$ asymptotically meets the upper bound of $r_q=e$ for general scaling functions. Intuitively, this can also be seen by constructing a piecewise linear function which approximates the function $q(y)=e^y$.
  
  We follow up with another result which restricts the space of phantom mechanisms with scaling when we require a bounded approximation ratio.
  \begin{lemma}\label{lem:simpphant}
      Let $f$ be a phantom mechanism with scaling, with phantoms $p_1(q),\dots,p_{n+1}(q)$. If there exists a scaling function $q$ such that $p_i(q)\neq 0$ and/or $p_i(q)\neq 1$ for all $i\in [n+1]$, then $f$ has an unbounded approximation ratio for total cost and maximum cost.
  \end{lemma}
  \begin{proof}
      Due to symmetry it suffices to consider a scaling function $q$ such that $p_i(q)\neq 1$ for all $i\in [n+1]$. If all agents are located at $1$, the optimal total and maximum cost are $0$, but $f$ places the facility at a location strictly left of $1$, as all $n+1$ phantoms are strictly left of $1$. Therefore the total cost and maximum cost approximation ratios are unbounded. \qed
  \end{proof}
 This lemma shows that when proving the approximation results, it suffices to consider phantom mechanisms with scaling where $p_1=0$ and $p_{n+1}=1$. This is equivalent to a phantom mechanism with scaling where there are only $n-1$ phantoms, and by Theorem~\ref{thm:charphant}, this characterizes the strategyproof and anonymous mechanisms which place the facility between the leftmost and rightmost agent locations.
		\subsection{Total Cost}
	We now prove approximation ratio results for the objective of total cost, a standard utilitarian measure of efficiency. We begin by giving a lower bound on the approximation ratio achievable by a phantom mechanism with scaling.
		\begin{theorem}\label{thm:totalcostgen}
		    There exists a scaling function $q$ such that any phantom mechanism with scaling has a total cost approximation ratio of at least $r_q$.
		\end{theorem}
  \begin{proof}
   By Lemma~\ref{lem:simpphant}, it suffices to consider phantom mechanisms with scaling where at least one phantom is on $0$ and at least one phantom is on $1$ (regardless of the scaling function), so for the remainder of the proof, we are concerned with the placement of the remaining $n-1$ phantoms.
  
  Consider an arbitrary phantom mechanism with scaling $f$, and a continuous scaling function $q$ which has global maxima at $0$, $0.5$, and $1$, and global minima at $0.25$ and $0.75$. We divide the remainder of the proof into two cases. In the first case, suppose that under $q$, $f$ places $k\in \{0,\dots,n-2\}$ phantoms in $[0,0.5]$ (and the remaining phantoms in $(0.5,1]$). Then for an agent location profile with even $n$, where $\frac{n}{2}$ agents are on $0$ and $\frac{n}{2}$ agents are on $0.5$, the facility will be placed on $0.5$. This results in a total cost of $\max_y q(y)\frac{n}{2}(0.5-0)$, but the optimal facility placement of $0.25$ gives a total cost of $\min_y q(y)*0.25n$. Dividing these terms gives us our lower bound of $r_q$. In the remaining case, suppose that under $q$, $f$ places all $n-1$ phantoms in $[0,0.5]$. Then we consider an agent location profile with even $n$ where $\frac{n}{2}$ agents are on $0.5$ and $\frac{n}{2}$ agents are on $1$. Here, the facility will be placed on $0.5$ to give a total cost of $\max_y q(y)\frac{n}{2}(1-0.5)$, but the optimal facility placement of $0.75$ gives a total cost of $\min_y q(y)*0.25n$. Again, dividing these terms gives our lower bound of $r_q$, completing the proof by exhaustion of cases. \qed
  \end{proof}
  From this result, we can also easily see that among all of the phantom mechanisms with scaling, simply placing the facility at the median agent gives the best possible approximation ratio for total cost.
  \begin{theorem}
      The median mechanism has a total cost approximation ratio of $r_q$.
  \end{theorem}
  \begin{proof}
      The median mechanism minimizes the sum of (unscaled) distances, so its total cost is at most $\max_y q(y) \sum_i |f_{med}(x)-x_i|$, whilst the optimal total cost is at least $\min_y q(y) \sum_i |f_{med}(x)-x_i|$. This shows the approximation ratio is at most $r_q$, and we have a matching lower bound from Theorem~\ref{thm:totalcostgen}. \qed
  \end{proof}
  Recall that phantom mechanisms with scaling characterize all strategyproof and anonymous mechanisms when agents have single-peaked preferences. We can consider scaling functions which ensure that agents have single-peaked preferences, and use Theorems~\ref{thm:genspe} and \ref{thm:piecelin} to find the specific total cost approximation ratio lower bounds for any strategyproof and anonymous mechanism.
  \begin{corollary}\label{cor: e}
  For arbitrary continuous scaling functions, when agents are guaranteed to have single-peaked preferences, any strategyproof and anonymous mechanism has a total cost approximation ratio of at least $e$.
  \end{corollary}
  \begin{corollary}\label{cor: 2}
  For piecewise linear scaling functions with $k$ line segments, when agents are guaranteed to have single-peaked preferences, any strategyproof and anonymous mechanism has a total cost approximation ratio of at least $(1+\frac{1}{k})^k$.
  \end{corollary}
  Again, the median mechanism would also give the best total cost approximation ratio, and can be suitably implemented for these restricted sets of scaling functions. However, in a scenario where the scaling function does not guarantee single-peaked agent preferences, phantom mechanisms with scaling no longer guarantee strategyproofness, and a constant mechanism has an unbounded approximation ratio. The \emph{dictator mechanism}, which selects a specific agent to be the dictator and places the facility at their reported location, is known to be strategyproof for general agent preferences (though it is clearly not anonymous). We show that this mechanism has a bounded approximation ratio for total cost.
		\begin{proposition}
			The dictator mechanism has a total cost approximation ratio of $(n-1)r_q$.
		\end{proposition}
		\begin{proof}
			Suppose there are $n-1$ agents at the same location away from the dictator. Also suppose that the scaling function is highest at the dictator's location, and lowest at the other $n-1$ agents' location. The total cost approximation ratio is therefore at least $(n-1)r_q$. To see that the approximation ratio cannot be higher than this, note that the total cost corresponding to the dictator mechanism is at most $d(n-1)\max_{y\in X}q(y)$ where $d\in (0,1]$, as the dictator receives $0$ cost, and the optimal total cost is at least $d\min_{y\in X}q(y)$ when not all agents are at the same point. \qed
		\end{proof}
		
		\subsection{Maximum Cost}
  We next prove approximation ratio results for the objective of maximum cost, a standard measure of egalitarian fairness. We begin by giving a lower bound for the maximum cost approximation ratio of any phantom mechanism with scaling.
  \begin{theorem}\label{thm:maxcostgenLB}
      There exists a scaling function $q$ such that any phantom mechanism with scaling has a maximum cost approximation ratio of at least $2r_q$.
  \end{theorem}
  \begin{proof}  
    The proof of this theorem is almost identical to the proof of Theorem~\ref{thm:totalcostgen}; we repeat the key arguments for completeness. By Lemma~\ref{lem:simpphant}, it suffices to consider mechanisms which always place at least one phantom on $0$ and at least one phantom on $1$, so we are concerned with the placement of the remaining $n-1$ phantoms.

    For the instance described in the proof of Theorem~\ref{thm:totalcostgen}, if at most $n-2$ phantoms are placed in $[0,0.5]$, the maximum cost from the facility placement at $0.5$ will be $\max_y q(y)(0.5-0)$, whilst the optimal maximum cost is $\min_y q(y)(0.5-0.25)$. By dividing these terms, we obtain our lower bound of $2r_q$. In the remaining case where all $n-1$ phantoms are placed in $[0,0.5]$, the maximum cost from the facility placement is $\max_y q(y) (1-0.5)$, whilst the optimal maximum cost is $\min_y q(y)(0.75-0.5)$. Again, dividing these terms gives our lower bound of $2r_q$, completing the proof. \qed
  \end{proof}

    Unlike the total cost objective, we find that under the maximum cost objective, any phantom mechanism with scaling which always places at least one phantom on $0$ and at least one phantom on $1$ has the same maximum cost approximation ratio which matches our lower bound.
    \begin{theorem}
        Any phantom mechanism with scaling which, for any $q\in \mathcal{Q}$, places at least one phantom on $0$ and at least one phantom on $1$, has a maximum cost approximation ratio of $2r_q$.
    \end{theorem}
  \begin{proof}
      By Theorem~\ref{thm:maxcostgenLB}, it suffices to show that the approximation ratio is at most $2r_q$. To see this, note that under a phantom mechanism with scaling which places at least one phantom on $0$ and at least one phantom on $1$, the facility is always placed within $[x_1,x_n]$. As a result, the mechanism's maximum cost is at most $(x_n-x_1)\max_{y\in X}q(y)$, whilst the optimal maximum cost is at least $\frac{x_n-x_1}{2}\min_{y\in X}q(y)$. \qed
  \end{proof}
  As in the previous section, we can also derive the maximum cost approximation ratios when the scaling functions are restricted such that agents have single-peaked preferences.
 \begin{corollary}\label{cor: 2e}
  For continuous scaling functions, when agents are guaranteed to have single-peaked preferences, any strategyproof and anonymous mechanism has a maximum cost approximation ratio of at least $2e$.
  \end{corollary}
   \begin{corollary}\label{cor: 4}
  For piecewise linear scaling functions with $k$ line segments, when agents are guaranteed to have single-peaked preferences, any strategyproof and anonymous mechanism has a maximum cost approximation ratio of at least $2(1+\frac{1}{k})^k$.
  \end{corollary}
  Finally, we prove the maximum cost approximation ratio for the dictator mechanism, which ensures strategyproofness for all scaling functions.
        \begin{proposition}
        The dictator mechanism has a maximum cost approximation ratio of $2r_q$.
        \end{proposition}
        \begin{proof}
        Consider an agent location profile where $x_1=x_2=\dots=x_{n-1}$ and $x_n\neq x_1$. When the scaling function has its minimum value at $\frac{x_1+x_n}{2}$ and its maximum value on the facility/dictator, we have $\frac{MC(q,f(q,\boldsymbol{x}),\boldsymbol{x})}{MC(q,y^*_{MC},\boldsymbol{x})}=2r_q$, showing that the approximation ratio is at least $2r_q$. For the matching upper bound, we note that the mechanism's maximum cost is at most $(x_n-x_1)\max_{y\in X}q(y)$ and the optimal maximum cost is at least $\frac{x_n-x_1}{2}\min_{y\in X}q(y)$. \qed
        \end{proof}
        
        \section{Discussion}
In this paper, we have initiated the study of facility location problems where the agents' costs include a multiplicative scaling factor which depends on where the facility is placed. This problem applies to many real-world scenarios where the facility's effectiveness depends on external factors influencing its location. Our results focus on the objectives of total cost and maximum cost, and on both continuous and piecewise linear scaling functions.

Although we focus on single-peaked agent preferences and have some mechanism design results for general agent preferences, the agent preferences in our setting are a strict subspace of general preferences, taking the form where the metric distance is multiplied by some continuous scaling function. An open question is to characterize the deterministic mechanisms satisfying strategyproofness and surjectivity under these preferences, which we conjecture to be dictatorships.

There are many natural extensions to our paper. We can consider multiple facilities and multiple dimensions, as well as randomized mechanisms. We could also restrict further the space of piecewise linear scaling functions to single-dipped or single-peaked linear scaling functions, and find even lower approximation ratio bounds. Finally, our concept of scaling functions can be applied to existing variations of the facility location problem, such as capacitated facilities \citep{ACLL+20} and obnoxious facility location \citep{CYZ13}.

\section*{Acknowledgements}
The authors would like to thank Matthias Greger for his extremely helpful comments and feedback. We also thank the Institute for Mathematical Sciences, National University of Singapore for facilitating our discussion with Matthias.

\newpage
\bibliographystyle{ACM-Reference-Format} 
\bibliography{citation}
		
\end{document}